\documentclass{lmcs}
\usepackage[utf8]{inputenc}

\usepackage{bussproofs,quiver,wrapfig}

\newtheorem{definition}{Definition}
\newtheorem{lemma}{Lemma}

\newtheorem{theorem}{Theorem}

\newtheorem{thex}{Example}

\newcommand{\lsem}{\mbox{$\lbrack\hspace{-0.3ex}\lbrack$}}
\newcommand{\rsem}{\mbox{$\rbrack\hspace{-0.3ex}\rbrack$}}
\newcommand{\sempar}[1]{\lsem #1 \rsem}

\newcommand{\true}{\top}
\newcommand{\false}{\bot}
\newcommand{\tvar}{\mathtt{tt}}
\newcommand{\fvar}{\mathtt{ff}}

\newcommand{\suc}{\mathrm{S}}
\newcommand{\pred}{\mathrm{P}}
\newcommand{\f}{\mathtt{f}}
\newcommand{\x}{\mathtt{x}}
\newcommand{\p}{p}

\newcommand{\ite}[3]{\mathtt{if}\ #1\ \mathtt{then}\ #2\ \mathtt{else}\ #3}
\newcommand{\ifthel}[3]{\mathtt{if}\ #1\ \mathtt{then}\ #2\ \mathtt{else}\ #3}

\newcommand{\maxc}{\mathrm{maxc}}
\newcommand{\minc}{\mathrm{minc}}

\renewcommand{\boolean}{\mathbf{bool}}
\newcommand{\indx}{\mathbf{ind}}

\newcommand{\mybox}[1]{\fbox{\parbox{1 \textwidth}{#1}}}

\newcommand{\set}[1]{\left\{#1\right\}}

\title{Rational functions via recursive schemes}
\author{Siddharth Bhaskar, Jane Chandlee, and Adam Jardine}
\date{\today}

\begin{document}

\maketitle

\section{Introduction}

One of the central notions of formal language theory is that of \emph{regularity}. The paradigmatic example of a regular language of finite strings over a finite alphabet has been generalized in several directions. There are now robust notions of regularity for other types of data (e.g., infinite strings and trees), as well as for functions as opposed to languages. Each notion of regularity can typically be characterized three ways: mechanically (via finite automata), algebraically (via grammars), and logically (via fragments of monadic second-order logic). Much of the enterprise of formal language theory has concerned itself with extending the notion of regularity and making it more robust.

The present paper is no exception. Here we are concerned with \emph{function classes} over finite strings. These are perhaps most easily approached through automata theory: a regular function is anything computable by a \emph{transducer}, a finite automaton whose transitions can print strings to output. However, life is immediately more complicated: the equivalence between deterministic, nondeterministic, one-way, and two-way automata are broken into three classes in the case of transducers:
\begin{enumerate}
    \item One-way deterministic transducers compute the class of \emph{subsequential} functions.\footnote{Sometimes called \emph{sequential}, cf. \cite{sakarovitch09}.} These are further factored into \emph{left-} and \emph{right-}subsequential depending on which way the transducer reads the string.
    
    \item One-way nondeterministic transducers compute the strictly larger class of \emph{rational} functions; these are closed under string reversal even though the transducer is one-way, and we do not have to factor into left and right.\footnote{Note that whenever considering nondeterministic transducers, we restrict ourselves to those that do compute functions, i.e., the output for a given input is invariant of the sequence of nondeterministic choices.}
    
    \item Two-way deterministic and nondeterministic transducers compute the yet strictly larger class of \emph{regular} functions. (These are obviously closed under string reversal.)
\end{enumerate}

Each of these classes admits a logical characterization as well. These are in the spirit of finite model theory: we interpret each string over a fixed alphabet $\Sigma$ as a finite structure over a fixed signature, also called $\Sigma$. The domain of the structure is the set of indices of the string, we have a predicate for each character in the alphabet that picks out those indices carrying that character, and we have some way of comparing or operating on indices, such as a linear order or successor and predecessor functions.

Under this identification of strings with finite structures, string languages can be identified with spectra. This is the fundamental bridge between automata-theoretic and logical characterizations of complexity classes. The prototypical such characterization is the identification of regular string languages with spectra of sentences in monadic second-order (MSO) logic \cite{buchi60}. This was extended by Engelfriet and Hoogeboom to regular \emph{functions} and MSO-definable \emph{interpretations} \cite{engelfriethoogeboom01}.

\paragraph{A word on interpretations}
Given two signatures $L$ and $K$, an interpretation $L \to K$ is a $K$-formula for every symbol of $L$. This gives a map from $K$-structures to $L$-structures: given a $K$-structure, we can interpret any $L$-symbol by its $K$-definition.\footnote{Note the reversal in direction from $K$ to $L$; properly formalized, we can define a contravariant semantics functor from interpretations to the maps they define.}
If $L$ and $K$ are string alphabets $\Sigma$ and $\Gamma$, then we can identify the set of finite $L$-structures and $K$-structures with $\Sigma^\star$ and $\Gamma^\star$ respectively, and an interpretation $L \to K$ defines a map $\Gamma^\star \to \Sigma^\star$.

In formal language theory, we are interested in defining functions which increase the string by at most a constant multiplicative factor.
Following the work of Engelfriet and Hoogeboom \cite{engelfriethoogeboom01} (and ultimately Courcelle \cite{courcelle94}), we identify indices of the defined $L$-string with single indices 
of the original $K$-string. To increase the constant factor beyond 1, we are allowed to do this a fixed finite number of times and merge the resulting copies. Henceforth, when we say \emph{interpretation}, we mean an interpretation in this sense.
\footnote{The term \emph{interpretation} comes from model theory, where the domain of the output structure is typically identified with \emph{tuples}, as opposed to copies, of the input structure. (For example, complex numbers can be identified with pairs of real numbers.) An alternate term used in formal language theory is \emph{transduction}.}

\paragraph{Order-preservation.}
Fix an interpretation $\pi$ from $\Sigma$ to $\Gamma$. Then if $s$ is some $\Gamma$-string and $|s|$ is its set of indices, $\pi$ defines a $\Sigma$-string whose set of indices is $m \times |s|$ for some fixed $m \in \omega$. There is a $\Gamma$-formula in $\pi$ which defines the order of these indices in the output string.
An \emph{order-preserving} interpretation is one which says: no, there is no such formula. Rather, you \emph{must} consider the indices of the output string in the natural lexicographic order on $m \times |s|$. 
Bojanczyk \cite{bojanczyk14} and Filiot \cite{filiot15} found that by refining MSO-interpretations to order-preserving MSO-interpretations, the resulting class of computable functions drops from regular to rational.

\paragraph{Boolean monadic recursive schemes}
The present authors \cite{bhaskaretal20} introduced the notion of a \emph{boolean monadic recursive scheme} (BMRS), a weak programming language on strings, and found that \emph{one-way} order-preserving BMRS interpretations compute exactly subsequential functions. BMRS are motivated by considerations in computational phonology, as they provide a formalism which simultaneously enforces the ``correct'' computational upper bound while being flexible enough that linguistically significant phenomena may be easily legible in the code \cite{CJ21}.

\paragraph{Our contributions}
The \emph{syntactic composition problem} is: given two interpretations computing functions $g$ and $f$, find an interpretation computing $g \circ f$. In the present paper we solve the syntactic composition problem for order-preserving BMRS interpretations.

Standard interpretations are naturally compositional. Given interpretations $\pi : L \to K$ and $\rho : K \to J$, we can form an interpretation $L \to J$ by substituting each occurrence of a $K$-symbol in the formulas of $\pi$ by the corresponding $J$-formula of $\rho$. 

Order-preserving interpretations, however, are \emph{not} naturally compositional: they use the index ordering on the input string, but do not define them on the output string. Compounding this, order-preserving interpretations are allowed to skip indices of the output on which no character has been defined.\footnote{So, even if the indices of the output are identified with $m \times |s|$, the output string might have length less than $m|s|$.} Naively, it seems like composing order-preserving interpretations would require some basic arithmetic to count skipped indices, whereas BMRS lack any mechanism for counting.\footnote{We write BMRS for both the singular and plural of a boolean monadic recursive scheme. This acronym has come to be pronounced \emph{beamers} in the plural and \emph{beamer} in the singular via back-formation.}

Our core technical contribution is a fine analysis of the syntactic composition problem for order-preserving interpretations.
We factor this syntactic composition problem into four parts, three of which are routine and go through in practically any logic. The last is difficult in the absence of counting and requires, it seems, precisely the computational power of a BMRS.

As a consequence of this analysis, we find that \emph{order-preserving BMRS interpretations compute the class of rational functions}. Any rational function can be factored as a composition of a left and right subsequential function \cite{ElgotMezei}. By our previous paper, each of these can be computed by a one-way BMRS; by the current paper, their composition is computed by a BMRS. In other words, while each BMRS interpretation is easily seen to be an MSO interpretation, we show that the converse is true as well.
This result, combined with that of our previous paper, shows that there are natural BMRS characterizations of both the subsequential and rational functions---a characterization not (currently) available to MSO logic.

\paragraph{Related work}
The monadic fragments of various programming languages already attracted attention from the early days of schematology, as they were often found to combine nontrivial expressive power while being more tractable than general recursive programs.
In formal language theory, monadic logics such as MSO occupy a place of central importance. The relationship between MSO and various monadic first-order logics equipped with a mechanism for recursion (such as monadic least fixed-point logic and monadic datalog) is an important one, and these are known to be as powerful as MSO over string and tree data, as least for boolean queries \cite{Sch06}.\footnote{At least one paper notes that such results are folklore in the database community, cf. \cite{GK04}.}

Despite the differences in our approach---our choice of string primitives, considering interpretations instead of simple queries, etc.---we surmise that our characterization of rational functions by order-preserving BMRS interpretations could probably be cobbled together out of known results. What we believe we have, however, is a genuinely new proof---one that comes out of solving the syntactic composition problem for order-preserving BMRS interpretations instead of a direct simulation of a given MSO interpretation.

\paragraph{Structure of this paper}

Section \ref{s:strings} gives a model-theoretic definition of strings, Section \ref{s:bmrs} defines BMRS over these structures, and Section \ref{sec:interpretations} defines order-preserving BMRS interpretations and shows how they define functions on strings.
The technical work is presented in Sections \ref{s:composition} and \ref{s:blank}, which show how to compose order-preserving BMRS interpretations; Section \ref{s:rational} uses this to establish that such interpretations capture rational functions.
Finally, Section \ref{s:conclusion} discusses further directions.

\section{Strings as finite structures}
\label{s:strings}

We are concerned with finite \emph{strings} over a finite \emph{alphabet}, or some nonempty set of symbols.
We typically use capital Greek letters (e.g., $\Sigma$, $\Gamma$, $\Delta$) to name alphabets, and lowercase Latin letters (e.g., $s$, $t$) to name strings.

Strings over an alphabet $\Sigma$ can be associated with a first-order \emph{signature}, also written $\Sigma$. 
\begin{definition}
Given an alphabet $\Sigma$, the signature $\Sigma$ consists of:
\begin{itemize}
    \item a monadic (i.e., unary) relation symbol $\sigma$ for each character $\sigma \in \Sigma$,
    \item monadic relation symbols $\max$ and $\min$, and
    \item monadic function symbols $\suc$ and $\pred$ (for \emph{successor} and \emph{predecessor}).
\end{itemize}
\end{definition}

Where the distinctions are necessary, we will refer to the elements of the alphabet $\Sigma$ as the \emph{characters} of $\Sigma$, and the functions and relations in the signature $\Sigma$ as the \emph{primitives} of $\Sigma$. 

We also identify each string in $\Sigma^\star$ with a finite $\Sigma$-structure as follows.

\begin{definition}
Given a string $s \in \Sigma^\star$, let the $\Sigma$-structure, also called $s$, have as a domain the set of indices of $s$ (which we represent with an initial segment of the natural numbers), and for an index $x$ of $s$, let
\begin{itemize}
    \item $s \models \sigma(x)$ iff character of $s$ at index $x$ is $\sigma$,
    
    \item $s \models \min(x)$ iff $x$ is the least index $0$,
    
    \item $s \models \max(x)$ iff $x$ is the greatest index $|s|-1$,
    
    \item $s \models \suc(x) = y$ iff $y = x + 1$ or $x = y = |s|-1$, and 
    
    \item $s \models \pred(x) = y$ iff $y = x - 1$ or $x = y = 0$.
\end{itemize}
\end{definition}

Note that successor and predecessor fix the greatest and least indices respectively. Note as well that a finite $\Sigma$-structure is a string if and only if for each $x < |s|$, $s \models \sigma(x)$ for a unique character $\sigma \in \Sigma$.
Without risk of ambiguity we often conflate strings in $\Sigma^\star$ with their first-order structure over the signature $\Sigma$.
Furthermore, we adopt the set-theoretic convention that identifies a natural number with its set of predecessors, i.e., $n = \{0, 1, \dots , n-1\}$.
This allows us to identify the domain of a string $s$ with its \emph{length} $|s|$.

\section{Boolean monadic recursive schemes}
\label{s:bmrs}

We now define \emph{boolean monadic recursive schemes} (BMRS), a programming language first introduced in \cite{bhaskaretal20}. These are presented in a pure functional style and equipped with a standard big-step environment-based semantics.\footnote{Our presentation is particularly influenced by McCarthy \cite{McCarthy59} via Moschovakis \cite{moschovakis19}.}

Programs are executed relative to a given finite string. Program variables range over two types of data: booleans $\boolean$ and string indices $\indx$. Each recursive function symbol in a boolean monadic recursive scheme is required to have type $\indx \to \boolean$; i.e., they are \emph{boolean}-valued and \emph{monadic} (i.e., have a single input variable), hence the name.

These twin requirements impose quite stringent limitations on the expressive power of boolean monadic recursive schemes. For example, our inability to program functions of type $\indx \to \indx$ means that we cannot do arithmetic on string indices to, e.g., locate the halfway point of a string. Our inability to program functions of type $\indx \times \indx \to \boolean$ means we cannot cheat by computing the graph relation of a function $\indx \to \indx$. The fact that inputs and outputs are of different types prohibits nested recursive calls.

We first define terms, then programs, and then their semantics. We use \texttt{typewriter} script for program syntax (i.e., program variables and keywords---the purely logical elements of our programming language). We fix a countably infinite set $\mathcal{F}$ of recursive function names, and a single variable $\x$ of type $\indx$.

\begin{definition}
Given a signature $\Sigma$, a \emph{$\Sigma$-term} is any term that can be derived from the inference rules in Figure \ref{fig:terms}.
\end{definition}

\begin{figure}
    \centering
    \mybox{
\begin{minipage}{.3 \textwidth}
\begin{prooftree}
\AxiomC{}
\UnaryInfC{$\x : \indx$}
\end{prooftree}
\end{minipage}
\begin{minipage}{.33 \textwidth}
\begin{prooftree}
\AxiomC{$T : \indx$}
\UnaryInfC{$\pred(T) : \indx$}
\end{prooftree}
\end{minipage}
\begin{minipage}{.33 \textwidth}
\begin{prooftree}
\AxiomC{$T : \indx$}
\UnaryInfC{$\suc(T) : \indx$}
\end{prooftree}
\end{minipage}
\\
    
\begin{minipage}{.3 \textwidth}
\begin{prooftree}
\AxiomC{}
\UnaryInfC{$\tvar : \boolean$}
\end{prooftree}
\end{minipage}
\begin{minipage}{.3 \textwidth}
\begin{prooftree}
\AxiomC{}
\UnaryInfC{$\fvar : \boolean$}
\end{prooftree}
\end{minipage}
\begin{minipage}{.33 \textwidth}
\begin{prooftree}
\AxiomC{$T : \indx$}
\AxiomC{$\f \in \mathcal{F}$}
\BinaryInfC{$\f(T) : \boolean$}
\end{prooftree}
\end{minipage}
\\

\begin{minipage}{.33 \textwidth}
\begin{prooftree}
\AxiomC{$T : \indx$}
\UnaryInfC{$\sigma(T) : \boolean$}
\end{prooftree}
\end{minipage}\begin{minipage}{.33 \textwidth}
\begin{prooftree}
\AxiomC{$T : \indx$}
\UnaryInfC{$\max(T) : \boolean$}
\end{prooftree}
\end{minipage}\begin{minipage}{.33 \textwidth}
\begin{prooftree}
\AxiomC{$T : \indx$}
\UnaryInfC{$\min(T) : \boolean$}
\end{prooftree}
\end{minipage}

\begin{center}
    \begin{prooftree}
    \AxiomC{$T_0 : \boolean$}
    \AxiomC{$T_1: \indx$}
    \AxiomC{$T_2 : \indx$}
    \TrinaryInfC{$\ifthel{T_0}{T_1}{T_2}: \indx$}
    \end{prooftree}
    
    \begin{prooftree}
    \AxiomC{$T_0 : \boolean$}
    \AxiomC{$T_1: \boolean$}
    \AxiomC{$T_2 : \boolean$}
    \TrinaryInfC{$\ifthel{T_0}{T_1}{T_2}: \boolean$}
    \end{prooftree}
\end{center}
    }
    \caption{$\Sigma$-terms. Here $\sigma$ ranges over the characters of $\Sigma$.}
    \label{fig:terms}
\end{figure}

\begin{definition}
A \emph{headless $\Sigma$-boolean monadic recursive scheme} is a finite set of lines
\begin{align*}
    \f_0(\x) &= T_0 \\
    \f_1(\x) &= T_1 \\
    &\vdots \\
    \f_k(\x) &= T_k
\end{align*}
such that for each $0 \le i \le k$, $T_i$ is a $\Sigma$-term which does not contain any other recursive function names other than $(\f_0,\dots,\f_k)$.\footnote{Note that this is a \emph{set} rather than a \emph{list} of lines; the order does not matter. However, overwhelming programming intuition compels us to present it like a list.}
\end{definition}

For a headless BMRS $p = \{ \f_i(\x) = T_i\}_{0 \le i \le k}$, by a \emph{$p$-term} we mean a term in which no recursive function name occurs besides $(\f_0,\dots,\f_k)$. In particular, each $T_i$ is a $p$-term.

\begin{definition}[Semantics]
Let $\Sigma$ be a signature, $s$ be a $\Sigma$-string, $x$ be an index of $s$, $\p$ be a headless $\Sigma$-BMRS, $T$ be a $p$-term, and $v$ be a value whose type agrees with $T$.
We define the five-place relation $s, x \vdash_\p T \to v$ according to the inference rules in Figure \ref{fig:semantics}.
\end{definition}

\begin{figure}
    \centering
    \mybox{
\begin{minipage}{.3 \textwidth}
\begin{prooftree}
\AxiomC{}
\UnaryInfC{$x \vdash \x \to x$}
\end{prooftree}
\end{minipage}
\begin{minipage}{.3 \textwidth}
\begin{prooftree}
\AxiomC{}
\UnaryInfC{$x \vdash \tvar \to \true$}
\end{prooftree}
\end{minipage}
\begin{minipage}{.3 \textwidth}
\begin{prooftree}
\AxiomC{}
\UnaryInfC{$x \vdash \fvar \to \false$}
\end{prooftree}
\end{minipage}

\begin{minipage}{.5 \textwidth}
\begin{prooftree}
\AxiomC{$x \vdash T \to v$}
\RightLabel{\scriptsize(if $v > 0$)}
\UnaryInfC{$x \vdash \pred(T) \to v-1$}
\end{prooftree}
\end{minipage}
\begin{minipage}{.5 \textwidth}
\begin{prooftree}
\AxiomC{$x \vdash T \to v$}
\RightLabel{\scriptsize(if $v = 0$)}
\UnaryInfC{$x \vdash \pred(T) \to v$}
\end{prooftree}
\end{minipage}

\begin{minipage}{.5 \textwidth}
\begin{prooftree}
\AxiomC{$x \vdash T \to v$}
\RightLabel{\scriptsize(if $v < |s| - 1$)}
\UnaryInfC{$x \vdash \suc(T) \to v+1$}
\end{prooftree}
\end{minipage}
\begin{minipage}{.5 \textwidth}
\begin{prooftree}
\AxiomC{$x \vdash T \to v$}
\RightLabel{\scriptsize(if $v = |s| - 1$)}
\UnaryInfC{$x \vdash \suc(T) \to v$}
\end{prooftree}
\end{minipage}

\begin{minipage}{.5 \textwidth}
\begin{prooftree}
\AxiomC{$x \vdash T \to v$}
\RightLabel{\scriptsize(if $v > 0$)}
\UnaryInfC{$x \vdash \min(T) \to \false$}
\end{prooftree}
\end{minipage}
\begin{minipage}{.5 \textwidth}
\begin{prooftree}
\AxiomC{$x \vdash T \to v$}
\RightLabel{\scriptsize(if $v = 0$)}
\UnaryInfC{$x \vdash \min(T) \to \true$}
\end{prooftree}
\end{minipage}

\begin{minipage}{.5 \textwidth}
\begin{prooftree}
\AxiomC{$x \vdash T \to v$}
\RightLabel{\scriptsize(if $v < |s| - 1$)}
\UnaryInfC{$x \vdash \max(T) \to \false$}
\end{prooftree}
\end{minipage}
\begin{minipage}{.5 \textwidth}
\begin{prooftree}
\AxiomC{$x \vdash T \to v$}
\RightLabel{\scriptsize(if $v = |s| - 1$)}
\UnaryInfC{$x \vdash \max(T) \to \true$}
\end{prooftree}
\end{minipage}
    
\begin{minipage}{.5 \textwidth}
\begin{prooftree}
\AxiomC{$x \vdash T \to v$}
\RightLabel{\scriptsize(if $s_v = \sigma$)}
\UnaryInfC{$x \vdash \sigma(T) \to \top$}
\end{prooftree}
\end{minipage}
\begin{minipage}{.5 \textwidth}
\begin{prooftree}
\AxiomC{$x \vdash T \to v$}
\RightLabel{\scriptsize(if $s_v \neq \sigma$)}
\UnaryInfC{$x \vdash \sigma(T) \to \false$}
\end{prooftree}
\end{minipage}

\begin{minipage}{1 \textwidth}
\begin{prooftree}
\AxiomC{$x \vdash T \to v$}
\AxiomC{$s,v \vdash T^\f \to w$}
\BinaryInfC{$x \vdash \f(T) \to w$}
\end{prooftree}
\end{minipage}

\begin{minipage}{1 \textwidth}
\begin{prooftree}
\AxiomC{$x \vdash T_0 \to \true$}
\AxiomC{$x \vdash T_1 \to v$}
\BinaryInfC{$x \vdash \ite{T_0}{T_1}{T_2} \to v$}
\end{prooftree}
\end{minipage}

\begin{minipage}{1 \textwidth}
\begin{prooftree}
\AxiomC{$x \vdash T_0 \to \false$}
\AxiomC{$x \vdash T_2 \to v$}
\BinaryInfC{$x \vdash \ite{T_0}{T_1}{T_2} \to v$}
\end{prooftree}
\end{minipage}
}
    \caption{Program semantics, $s$ and $\p$ omitted for legibility. For a recursive function name $\p$, $T^\f$ is its recursive definition in $\p$, and for an index $v$, $s_v$ is the character of $s$ at $v$.}
    \label{fig:semantics}
\end{figure}

Finally, we give our boolean monadic recursive schemes heads. Unlike the recursive programs of \cite{moschovakis19}, we allow for the possibility of multiheaded programs.

\begin{definition}
For a signature $\Sigma$, a $\Sigma$-boolean monadic recursive scheme is composed of a headless $\Sigma$-boolean monadic recursive scheme $p$ along with a finite nonempty collection of $p$-terms, called the \emph{heads}.
\end{definition}

Following \cite{moschovakis19}, given a $\Sigma$-BMRS, we call the underlying headless part its \emph{body}. We extend the usage of $s,x \vdash_p T \to v$ to include $p$ which have heads; in this case, remember that the meaning of $\vdash$ does not depend on them.

Finally, we define two important subclasses of programs.

\begin{definition}
A $\mathrm{BMRS}^\pred$, or \emph{predecessor} BMRS, is one in which the successor function $\suc$ does not occur. Similarly a $\mathrm{BMRS}^\suc$, or \emph{successor} BMRS, is one in which the predecessor function $\pred$ does not occur.
\end{definition}

\section{Interpretations}\label{sec:interpretations}

\begin{definition}
Let $\Sigma$ and $\Gamma$ be alphabets. An ($m$-fold) interpretation $\pi : \Sigma \times m \to \Gamma$ is a multiheaded $\Gamma$-BMRS $\pi$ with a head $\pi(\sigma,i)$ for each character $\sigma \in \Sigma$ and $i < m$.
\end{definition}

\begin{definition}
An interpretation $\pi : \Sigma \times m \to \Gamma$ is \emph{well-defined} in case for each string $s \in \Gamma^\star$, $i < m$, and $x < |s|$, $s,x \vdash \pi(\sigma,i) \to \true$ for at most one $\sigma \in \Sigma$ and $s,x \vdash \pi(\tau,i) \to \false$ for every other $\tau \in  \Sigma$.

It is additionally \emph{strict} if $s,x \vdash  \pi(\sigma,i) \to \true$ for exactly one $\sigma$.
\end{definition}

\begin{definition}\label{def:transduction}
If $\pi : \Sigma \times m \to \Gamma$ is a well-defined interpretation, then the \emph{transduction induced by $\pi$} is the function $ \sempar{\pi}: \Gamma^\star \to \Sigma^\star$ such that for every $s \in \Gamma^\star$, 
\begin{itemize}
    \item $|\sempar{\pi}(s)| = |J|$, where $J  = \{(q,r) \in |s| \times m: (\exists \sigma \in \Sigma)\, s,q \models \pi(\sigma,r) \to \true \}$, and
    
    \item for each $x < |J|$, $\sempar{\pi}(s) \models \sigma(x) \iff s,q \models \pi(\sigma,r) \to \true$, where $(q,r)$ is the unique element of $J$ with $x$ predecessors, where
    
    \item $|s| \times m$ is ordered lexicographically, with $|s|$ being the more significant and $m$ being the less significant coordinate.
\end{itemize}
\end{definition}
In the important special case that $\pi$ is strict, this specializes to:
\begin{itemize}
    \item $|\sempar{\pi}(s)| = m |s|$, for every $s \in \Gamma^\star$, and
    
    \item for each $x < m|s|$, $ \sempar{\pi}(s) \models \sigma(x) \iff  s,q \vdash \pi(\sigma,r) \to \true,$ where $r$ and $q$ are the remainder and quotient respectively of $x \div m$.
\end{itemize}


For example, let $\Gamma=\set{a}$, $\Sigma={7,8,9}$, and $m=3$.
The interpretation $\pi$ where $\pi(7,0)=\pi(8,1)=\pi(9,2)=a(\x)$ defines a function that takes a string $s$ of $|s|$ $a$ characters and returns a string of length $|s|\times 3$ of the form $(789)^{|s|}$.
To see how this obtains, the following table shows how each index in the string $aaaa$ is interpreted at each $m$-coordinate $i$: a $7$ at $m$-coordinate $0$, an $8$ at $m$-coordinate $1$, and a $9$ at $m$-coordinate $2$. 
\[
  \begin{array}[t]{lllll}
     & 0 & 1 & 2 & 3 \\
   i & a & a & a & a \\\hline
   0 & 7 & 7 & 7 & 7 \\
   1 & 8 & 8 & 8 & 8 \\
   2 & 9 & 9 & 9 & 9 \\
  \end{array}
\]
By collating the copies of the indices of $aaaa$ first by order of index and then by order of $i<3$, we obtain the string $789789789789$. 
Note that $\pi$ is strict.

The authors have previously shown \cite{bhaskaretal20} that order-preserving $\mathrm{BMRS}^\pred$ interpretations and order-preserving $\mathrm{BRMS}^\suc$ interpretations describe exactly the left-subsequential functions and right-subsequential functions, viz.,

\begin{theorem}
  \label{thm:subseq}
  For any well-defined order-preserving $\mathrm{BMRS}^\pred$ (resp. $\mathrm{BMRS}^\suc$) interpretation $\pi$, $\sempar{\pi}$ is a left-subsequential (resp. right-subsequential) function. 
  Likewise, for any left-subsequential (resp. right-subsequential) function $f$, $f=\sempar{\pi}$ for some order-preserving $\mathrm{BMRS}^\pred$ (resp. $\mathrm{BMRS}^\suc$) interpretation $\pi$. 
\end{theorem}

This paper solves one of the questions left open in \cite{bhaskaretal20}; that is, to characterize order-preserving BMRS interpretations in the presence of both $\pred$ and $\suc$. It turns out that we get this more or less for free from the solution of the syntactic composition problem.

Before moving on, we note that there is a natural notion of \emph{substitution} of an interpretation into a program.

\begin{definition}\label{def:substitution}
Suppose $p$ is a headless $\Sigma$-BMRS and $\pi : \Sigma \times m \to \Gamma$ is an interpretation. For each $i < m$, let $p^\pi_i$ be the $\Gamma$-BMRS obtained from $p$ by replacing any occurrence of $\sigma$ with $\pi(\sigma,i)$.
\end{definition}

As stated this definition is a bit imprecise. Since the BMRS language contains no \texttt{let-} or \texttt{where-} constructs that allow you to directly embed a program into another program, what we mean is this: take $p$, stick a copy of the body of $\pi$ underneath, and replace each occurrence of $\sigma$ by the head corresponding to $\pi(\sigma,i)$. We trust that this is sufficiently clear.\footnote{A  technical note: when we take the union of two headless programs, we possibly rename recursive function names so there are no accidental overlaps.} Note that every $p$-term is a $p^\pi_i$-term.

Substitutions enjoy the following property.

\begin{lemma}\label{lem:fundamental adjunction}
For any well-defined strict interpretation $\pi : \Sigma \times m \to \Gamma$, any string $s \in \Gamma^\star$, any headless $\Sigma$-BMRS $p$, any $p$-term $T$, any $x < m|s|$, and any boolean $b \in \{\true,\false\}$, 
$$ \sempar{\pi}(s),x \vdash_p T \to b \iff s,q \vdash_{p^\pi_r} T \to b,$$
where $q < |s|$ and $r < m$ are the quotient and remainder respectively of $x \div m$.
\end{lemma}

\section{Syntactic composition}
\label{s:composition}

We now turn to the following question: given two ``composable'' interpretations $\pi : \Sigma \times m \to \Gamma$ and $\rho : \Delta \times n \to \Sigma$, how do we define an interpretation $\omega : \Delta \times mn \to \Gamma$ such that $\sempar{\rho} \circ \sempar{\pi} = \sempar{\omega}$?
This is an instance of a general problem in the theory of programming languages, namely: how do we pull back a given semantic operation on functions to a corresponding syntactic operation on the program texts themselves?\footnote{Cf. \cite{Jones91} where such operations are called \emph{symbolic} rather than \emph{syntactic}.} 

In most programming languages, the fact that programs \emph{can} be composed is not in and of itself hard to show; the subtlety (if there is one) usually lies in how efficient the composition can be made. However, this is not the case for order-preserving BMRS interpretations. For example, suppose we wanted to compose the interpretations $\rho$ and $\pi$ above. Loosely speaking, $\rho$ defines each function $\delta(x)$ (``index $x$ carries character $\delta$''), for $\delta \in \Delta$, using $\sigma(x)$ as primitives, for $\sigma \in \Sigma$. Similarly $\pi$ defines each function $\sigma(x)$, for $\sigma \in \Sigma$, using $\gamma(x)$ as primitives, for $\gamma \in \Gamma$.

So one is tempted to define $\omega$ by taking $\rho$ and replacing each call to $\sigma(x)$ by its definition in $\pi$. That is indeed the right idea. The problem is that order-preserving interpretations are not \emph{quite} compositional: they use the successor and predecessor functions on the input strings, but they are not required to define them on their output string.

This is not a huge problem for strict interpretations, in which index successors and predecessors easily carry over from input to output strings. (For example, in a strict 1-fold interpretation we can simply identify the indices in the input and output strings.)
Life gets more complicated, however, for non-strict interpretations, which may \emph{skip} certain indices in the output string. Naively, it would seem that to relate indices of the input and output, a program would have to \emph{count} skipped indices, an impossibility in languages with boolean-valued recursive functions. Overcoming this obstacle is the core technical contributions of this paper.

In this section, we factor the problem of composing functions into several sub-problems. The syntactic versions of all but one of these sub-problems are routine, and we dispatch them in the present section. The final sub-problem cuts to the heart of the matter, so we postpone it to its own section.

\subsection{Four sub-problems}
Say that a \emph{sharply bounded function} is a string function $f$ such that there exists a natural number $n$ satisfying $|f(x)| \le n|x|$, for each $x$ in the domain of $f$. In this case, we say $f$ is \emph{sharply bounded} by $n$. Of course, if $\pi$ is an $n$-fold interpretation, then $\sempar{\pi}$ is sharply bounded by $n$. By analogy to interpretations, say that a sharply bounded function $f$ is \emph{strict} if there exists an $n$ such that $|f(x)| = n|x|$ for each $x$ in the domain of $f$.

The problem we want to consider is: given two sharply bounded functions $f : \Gamma^\star \to \Sigma^\star$ and $g : \Sigma^\star \to \Delta^\star$ bounded by $n$ and $m$ respectively, how can we compose them to obtain $g \circ f : \Gamma^\star \to \Delta^\star$ bounded by $nm$? Let's say we know how to compose \emph{strict} sharply bounded functions, and we want to reduce the general problem to the strict case in the simplest way possible.

Suppose $\square$ is a character that appears neither in $\Sigma$, $\Delta$, nor $\Gamma$. (We shall think of it as a ``blank'' character.) Let $\Sigma_\square$, $\Delta_\square$, and $\Gamma_\square$ be obtained from $\Sigma$, $\Delta$, and $\Gamma$ respectively by adding this blank character. Then there is a natural \emph{deletion} map, e.g., $d_\Sigma : \Sigma^\star_\square \to \Sigma^\star$ that deletes all blank characters from a given string. (So, e.g., $d(ab\square a \square \square) = aba$) We shall abuse notation and just write $d$ for $d_\Sigma$, $d_\Gamma$, etc.

Say that $f' : \Gamma^\star \to \Sigma^\star_\square$ is a \emph{strictification} of $f : \Gamma^\star \to \Sigma^\star$ in case both functions are sharply bounded by the same bound, $f'$ is strict, and $d \circ f' = f$. Any sharply bounded function admits a (non-unique) strictification by arbitrarily padding each $f(x)$ with blanks until its length is equal to $n |x|$. We might hope that we can obtain $g \circ f$ from sharply bounded functions $f$ and $g$ by

\begin{itemize}
    \item finding two strictifications $f'$ and $g'$ of $f$ and $g$, and
    \item composing them (i.e., $g' \circ f'$).
\end{itemize}

This has two problems. Less seriously, the codomain of $g'$ is $\Delta_\square^\star$ and not $\Delta^\star$. We fix this by simply composing on the outside by $d : \Delta^\star_\square \to \Delta^\star$, an easy operation on program codes. More seriously, $f'$ and $g'$ are not composable, as the codomain of $f'$ and the domain of $g'$ are not identical ($f' : \Gamma^\star \to \Sigma^\star_\square$ while $g' : \Sigma^\star \to \Delta^\star_\square$). Naively we might try to fix this by replacing $f'$ by $d \circ f'$, until we realized that we are now tasked with composing the (non-strict) $d \circ f'$ with $g$, placing us right back where we started.

Instead, we convert $g' : \Sigma^\star \to \Delta_\square^\star$ into another  strict function $g_b : \Sigma^\star_\square \to \Delta^\star_\square$ such that $d \circ g_b  = d \circ g' \circ d$, i.e., the following diagram commutes:
\[\begin{tikzcd}
	{\Sigma^\star_\square} & {} & {\Delta^\star_\square} \\
	&&& {\Delta^\star} \\
	{\Sigma^\star} & {} & {\Delta^\star_\square }
	\arrow["d"', from=1-1, to=3-1]
	\arrow["d"', from=3-3, to=2-4]
	\arrow["{g'}"{description}, from=3-1, to=3-3]
	\arrow["{g_b}"{description}, from=1-1, to=1-3]
	\arrow["d", from=1-3, to=2-4]
\end{tikzcd}\]
In which case,
$$ d \circ g_b \circ f' = d \circ g' \circ d \circ f' = g \circ f$$
which is what we wanted to compute in the first place. Diagrammatically,
\[\begin{tikzcd}
	&& {\Sigma^\star_\square} & {} & {\Delta^\star_\square} \\
	&&&& {\Delta^\star_\square } && {\Delta^\star} \\
	{\Gamma^\star} && {\Sigma^\star} & {}
	\arrow["g", curve={height=24pt}, from=3-3, to=2-7]
	\arrow["d"', from=1-3, to=3-3]
	\arrow["d"', from=2-5, to=2-7]
	\arrow["{g'}"{description}, from=3-3, to=2-5]
	\arrow["f"{description}, from=3-1, to=3-3]
	\arrow["{f'}"{description}, from=3-1, to=1-3]
	\arrow["{g_b}"{description}, from=1-3, to=1-5]
	\arrow["d", curve={height=-6pt}, from=1-5, to=2-7]
\end{tikzcd}\]
Therefore, we have factored the problem of composing $g$ with $f$ into the following subtasks:
\begin{enumerate}
    \item Strictifying $f$ into $f'$ and $g$ into $g'$.
    \item Lifting $g'$ to $g_b$, which we call \emph{blank-enrichment}.
    \item Composing strict functions $g_b$ and $f'$.
    \item De-strictifying $(g_b \circ f')$ into $d \circ (g_b \circ f')$.
\end{enumerate}
Of these, strictifying and de-strictifying are the simplest to realize as operations on program codes, and we deal with them first. Composition of strict interpretations is straightforward, thought it takes some care to state and verify cleanly. Blank-enrichment is the most complicated, requiring both new ideas and attention to technical detail, and it is this which we postpone to the end.

\subsection{Strictification}

This is perhaps the easiest transformation on program codes. Given an interpretation $\pi : \Sigma \times m \to \Gamma$, we define $\pi' : \Sigma_\square \times m \to \Gamma$ by simply saying that an index of the input string carries the blank character $\square$ in $\pi'$ when that index carries no character in $\pi$. This has the effect that $\sempar{\pi'}$ is obtained by stuffing $\square$ where $\sempar{\pi}$ had nothing, effectively strictifying it. More precisely:

\begin{definition}
Given $\pi : \Sigma \times m \to \Gamma$, let $\pi' : \Sigma_\square \times m \to \Gamma$ have the same body as $\pi$. Define its heads by:
\begin{itemize}
    \item $\pi'(\sigma,i) \equiv \pi(\sigma,i)$ for each $i$ and $\sigma \neq \square$, and
    
    \item $\pi'(\square,i) \equiv \bigwedge_{(\sigma \neq \square)} \neg \pi(\sigma,i)$.\footnote{This should be regarded as syntactic sugar for a list of if-then-else statements.}
\end{itemize}
\end{definition}

Then we have:

\begin{lemma}\label{lem:strictification correctness}
If $\pi$ is well-defined, then $\pi'$ is well-defined and strict; moreover, $d \circ \sempar{\pi'} = \sempar{\pi}$.
\end{lemma}
\begin{proof}
Fix $s \in \Gamma^\star$, $i \in m$, and $x < |s|$. If $s,x \vdash \pi(\sigma,i) \to \true$ for some $\sigma \in \Sigma$ then $s,x \vdash \pi'(\sigma,i) \to \true$ for just that $\sigma$, and $s,x \vdash \pi'(\square,i) \to \false$. If $s,x \vdash \pi(\sigma,i) \to \true$ for no $\sigma \in \Sigma$, then the same is true of $\pi'$; moreover, $s,x \vdash \pi'(\square,i) \to \true$. This shows that $\pi'$ is well-defined and strict.

To show that $d \circ \sempar{\pi'} = \sempar{\pi}$, it suffices to show that for every $\Gamma$-string $s$ and index $x$ of $\sempar{\pi}(s)$, if $\sempar{\pi}(s)$ carries $\sigma$ at index $x$, then $\sempar{\pi'}(s)$ carries $\sigma$ at index $y$, where $y$ is the index of $\sempar{\pi'}(s)$ with $x$ non-blank predecessors.

Fix a string $s \in \Gamma^\star$. The indices of $\sempar{\pi}(s)$ can be identified with the set $J$ of pairs $(q,r)\in |s|\times m$ such that $s,q \models \pi(\sigma,r) \to \true$ for some $\sigma \in \Sigma$. Since $\pi'$ is strict, the indices of $\sempar{\pi}(s)$ can be identified simply with the set $|s| \times m$. Order this set lexicographically, first on $q$ then on $r$, and let $J$ inherit the induced order as a subset. 

Now fix $x < |\sempar{\pi}(s)|$. There is a unique $\sigma \in \Sigma$ such that $\sempar{\pi}(s) \models \sigma(x)$. Let $(q,r)$ be the unique element of $J$ with $x$ predecessors. Then $s,q \models \pi(\sigma,r) \to \true$, so $s,q \models \pi'(\sigma,r) \to \true$, which says that $\sempar{\pi'}(s) \models \sigma(mq + r)$. But $mq + r$ is exactly the index of $\sempar{\pi'}(s)$ with $x$ non-blank predecessors, which is exactly what we wanted to show.
\end{proof}

\subsection{De-strictification}

Suppose we have a well-defined, strict interpretation $\pi : \Sigma_\square \times m \to \Gamma$, and we want to obtain a well-defined, non-strict interpretation $\pi^\dag : \Sigma \times m \to \Gamma$ such that $\sempar{\pi^\dag} = d \circ \sempar{\pi}$. This transformation is extremely simple: we simply take $\pi$ and ``forget'' each $\pi(\square,i)$.

\begin{definition}
Given $\pi$ as above, let $\pi^\dag$ have the same body. For each $i<m$ and character $\sigma \in \Sigma$,  let
$$ \pi^\dag(\sigma,i) \equiv \pi(\sigma,i) .$$
\end{definition}

Then we have:
\begin{lemma}\label{lem:destrictification correctness}
If $\pi$ is a well-defined strict interpretation, then $\pi^\dag$ is well-defined; moreover, $\sempar{\pi^\dag} = d \circ \sempar{\pi}$.
\end{lemma}
\begin{proof}
To show that $\pi^\dag$ is well-defined we don't even need that $\pi$ is strict; $\pi^\dag$ trivially inherits well-definedness from $\pi$. 

As in the proof of Lemma \ref{lem:strictification correctness}: to show that $d \circ \sempar{\pi} = \sempar{\pi^\dag}$, it suffices to show that for every $\Gamma$-string $s$ and index $x$ of $\sempar{\pi^\dag}(s)$, if $\sempar{\pi^\dag}(s)$ carries $\sigma$ at index $x$, then $\sempar{\pi}(s)$ carries $\sigma$ at index $y$, where $y$ is the index of $\sempar{\pi}(s)$ with $x$ non-blank predecessors. As before, let $J$ be those pairs $(q,r)$ in $|s| \times m$ such that $s,q \vdash \pi(\sigma,r) \to \true$ for some $\sigma \in \Sigma$.

Fix $s \in \Gamma^\star$ and an index $x$ of $\sempar{\pi^\dag}(s)$. Let $\sigma$ be the character of $\sempar{\pi^\dag}(s)$ at $x$, so that $s,q \vdash \pi(\sigma,r) \to \true$, where $(q,r)$ is the unique element of $J$ with $x$ predecessors. But then $mq + r$ is the index of $\sempar{\pi}(s)$ with $x$ non-blank predecessors, which is what we wanted to show.
\end{proof}

\subsection{Strict composition}

Suppose that we have two strict interpretations $\pi : \Sigma \times m \to \Gamma$ and $\rho : \Delta \times n \to \Sigma$. How can we construct a strict interpretation $\mu : \Delta \times mn \to \Gamma$ such that $\sempar{\mu} = \sempar{\rho} \circ \sempar{\pi}$?

Consider the following example for motivation. Suppose that $\Gamma = \{0,1\}$, $\Sigma = \{a,b,c\}$, $m = 2$, $\Delta = \{8,9\}$, and $n = 3$. Suppose that $\sempar{\pi}(010)=abbcaa$, and $\sempar{\rho}(abbcaa) = 988989998 998 998 899$ i.e.,

\begin{minipage}{0.3\textwidth}
\centering
\begin{tabular}{c|c|c|c}
     & 0 & 1 & 0  \\ \hline
     $c_0$ & $a$ & $b$ & $a$  \\ \hline
    $c_1$ & $b$ & $c$ & $a$ 
\end{tabular}
\end{minipage}
\begin{minipage}{0.7\textwidth}
\centering
\begin{tabular}{c|c|c|c|c|c|c}
     & $a$ & $b$ & $b$ & $c$ & $a$ & $a$  \\ \hline
    $c_0$ & 9 & 9 & 9 & 9 & 9 & 8  \\ \hline
    $c_1$ & 8 & 8 & 9 & 9 & 9 & 9 \\ \hline 
    $c_2$ & 8 & 9 & 8 & 8 & 8 & 9
\end{tabular}
\end{minipage}

Then we want $\sempar{\mu}(010) = 988989 998998 998899$, i.e.,

\begin{center}
    \begin{tabular}{c|c|c|c}
     & 0 & 1 & 0  \\ \hline
    $c_0$ & 9 & 9 & 9  \\ \hline 
    $c_1$ & 8 & 9 & 9 \\ \hline
    $c_2$ & 8 & 8 & 8  \\ \hline
    $c_3$ & 9 & 9 & 8 \\ \hline
    $c_4$ & 8 & 9 & 9  \\ \hline
    $c_5$ & 9 & 8 & 9 
\end{tabular}
\end{center}

We can form this new table from the two old ones like so:
\begin{center}
    \begin{tabular}{c||c||c||c}
     & 0 & 1 & 0  \\ \hline \hline
    $c_0$ & 9 & 9 & 9  \\ \hline 
    $c_1$ & 8 & 9 & 9 \\ \hline 
    $c_2$ & 8 & 8 & 8  \\ \hline \hline
    $c_3$ & 9 & 9 & 8 \\ \hline
    $c_4$ & 8 & 9 & 9  \\ \hline
    $c_5$ & 9 & 8 & 9 
\end{tabular}
\end{center}
Each double-edge box is a ``tile.'' We take the first table and replace each $a$, $b$, and $c$ by the appropriate tile obtained by the second table. Out of the 6 copies $c_0,\dots,c_5$ in the composed table, the quotient upon division by $3$ tells us which copy ($c_0$ or $c_1$) to look at in the first table, and the remainder ($c_0$, $c_1$, or $c_2$) tells us which copy to look up in the second table. This suggests the following definition of composition.

\begin{definition}
Given an interpretation $\pi : \Sigma \times m \to \Gamma$ and $\rho : \Delta \times n \to \Sigma$, define an interpretation $(\rho \otimes \pi) : \Delta \times mn \to \Gamma$ as follows:
\begin{itemize}
    \item The body of $\rho \otimes \pi$ is the union of the bodies of $\rho^\pi_q$, for each $q \in m$.\footnote{Another technical note: each $\rho^\pi_q$ contains a copy of the body of $\pi$. These can be identified in the union. Otherwise, the recursive function symbols in each component must be renamed so they are not identified in the union.}
    
    \item For every $i \in mn$, let $q$ and $r$ be the integer quotient and remainder respectively of $i \div n$. (Then $r \in n$ and $q \in m$.) For every $(\delta,i) \in \Delta \times mn$, let the head $(\rho \otimes \delta)(\delta,i)$ be $\rho(\delta,r)^\pi_q$, i.e., the term obtained from $\rho(\delta,r)$ by replacing each occurrence of any $\sigma$ by $\pi(\sigma,q)$.
\end{itemize}

\end{definition}

First observe that this transformation preserves strictness: if $\pi$ and $\rho$ are well-defined and strict, then so is $\rho \otimes \pi$.
Next, we prove correctness:

\begin{lemma}\label{lem:strict comp correctness}
$\sempar{\rho \otimes \pi} = \sempar{\rho} \circ \sempar{\pi}$.
\end{lemma}

\begin{proof}
We know that the strings $\sempar{\rho \otimes \pi}(s)$ and $\sempar{\rho} ( \sempar{\pi}(s))$ each have length $mn |s|$. For each $x < mn|s|$ and $\delta \in \Delta$, we must show that
$$ \sempar{\rho \otimes \pi}(s) \models \delta(x) \iff \sempar{\rho}(\sempar{\pi}(s)) \models \delta(x).$$
Let $q < |s|$ and $r < mn$ be the quotient and remainder respectively of $x \div mn$. Then $x = mn \cdot q + r$. Let $q' < m$ and $r' < n$ be the quotient and remainder respectively of $r$ upon division by $n$. Then $r = n \cdot q' + r'$, and $x = n \cdot (m \cdot q + q') + r'$. Hence $r'$ is also the remainder of $x \div n$, and the quotient $q^\dagger$ of $x \div n$ is $m \cdot q + q'$. Observe that the quotient and remainder of $q^\dagger \div m$ are $q$ and $q'$ respectively.

By definition of $\sempar{\rho \otimes \pi}$, 
$$\sempar{\rho \otimes \pi}(s) \models \delta(x) \iff s,q \vdash (\rho \otimes \pi)(\delta,r) \to \true.$$
By definition of $\otimes$, $(\rho \otimes \pi)(\delta,r)$ is $ \rho(\delta,r')^\pi_{q'}$, so
$$\sempar{\rho \otimes \pi}(s) \models \delta(x) \iff s,q \vdash \rho(\delta,r')^\pi_{q'} \to \true.$$

On the other hand, 
$$\sempar{\rho}(\sempar{\pi}(s)) \models \delta(x) \iff \sempar{\pi}(s),q^\dag \vdash \rho(\delta,r') \to \true,$$
by definition of $\sempar{\rho}$. By Lemma \ref{lem:fundamental adjunction}, $$\sempar{\pi}(s),q^\dag \vdash \rho(\delta,r') \to \true \iff s,q \vdash \rho(\delta,r')^\pi_{q'} \to \true.$$

Hence, by composing these equivalences,
$$ \sempar{\rho \otimes \pi}(s) \models \delta(x) \iff \sempar{\rho}(\sempar{\pi}(s)) \models \delta(x), $$
which is what we wanted to show.
\end{proof}

\section{Blank enrichment}
\label{s:blank}

We now tackle the following problem: given a strict interpretation $\pi : \Sigma_\square \times m \to \Gamma$, find a strict interpretation $\pi_b : \Sigma_\square \times m \to \Gamma_\square$ such that $d \circ \sempar{\pi_b} = d \circ \sempar{\pi} \circ d$. The following terminology is helpful in understanding how this works. For a string $s \in \Gamma_\square^\star$, call $d(s)$ its \emph{underlying} $\Gamma^\star$-string. Conversely, call any member of the $d$-preimage of a $\Gamma^\star$-string $s$ a \emph{padding} of $s$. Call two strings $s,t \in \Gamma_\square^\star$ \emph{siblings} if $d(s) = d(t)$, i.e., they have the same underlying string.

Now what we want is to construct $\pi_b$ such that for any string $s \in \Gamma_\square^\star$, $\sempar{\pi_b}(s)$ and $\sempar{\pi}(t)$ are siblings, where $t$ is the underlying string of $s$. Of course, there are many ways to do this. We do so in the most straightforward way possible. For example, if $m=1$, we simply reproduce $\square$ in $\sempar{\pi_b}(s)$ wherever $s$ has a blank, then fill in $\sempar{\pi}(t)$ in the remaining spaces. For example, let $\Gamma = \{0,1\}$, $\Sigma_\square = \{a,b,\square\}$. Then we would want to define $\pi_b$ such that:
\begin{itemize}
    \item if $\sempar{\pi}(00) = ab$, $\sempar{\pi_b}(00 \square)  = ab \square$;
    
    \item if $\sempar{\pi}(00) = ab$, $\sempar{\pi_b}(0 \square0)  = a\square b$;

    \item if $\sempar{\pi}(01) = b \square$, $\sempar{\pi_b}(0 \square 1) = b \square \square$,
    
    \item if $\sempar{\pi}(01) = b \square$, $\sempar{\pi_b}( \square 0 1) = \square b \square$,

    \item if $\sempar{\pi}(010011) = aabbab$, $\sempar{\pi_b}(01 \square 00 \square 11) = aa \square bb \square ab$;
    
    \item if $\sempar{\pi}(010011) = aabbab$, $\sempar{\pi_b}(01 0 \square 0 1 \square \square 1) = aab \square ba \square \square b$.
\end{itemize}
Notice that in all of these instances $\sempar{\pi_b}(s)$ and $\sempar{\pi}(t)$ are siblings, where $t$ is the underlying string of $s$.

When $m = 2$, we do the same thing, but reproduce 2 copies of $\square$ in the output for each $\square$ in the input. For example, if $\sempar{\pi}(010) = aa\square bab$, we would like $\sempar{\pi_b}( \square 01 \square 0) = \square \square aa \square b \square \square ab$. The $\square \square$ pattern in indices $(0,1)$ and $(6,7)$ in the output come from the $\square$'s at indices $0$ and $3$ in the input. The rest comes from filling in $\sempar{\pi}(010)$ in the remaining indices.

Why is realizing this transformation on program codes hard? Again, briefly assume that $m=1$.\footnote{All the technical difficulty is encapsulated in the case $m=1$, but it is less cumbersome to discuss.} The basic idea is that to figure out which character is carried by index $x$ in $\sempar{\pi_b}(s)$, we have to figure out the number $x^\star$ of \emph{non-blank predecessors} of $x$ (that is, the number of indices of $\sempar{\pi_b}(s)$ preceding $x$ which do not carry $\square$), and look up index $x^\star$ in the string $\sempar{\pi}(s)$.

This sounds easy: given $x$, all we have to do is compute $x^\star$ and use $\pi$ to find the character it carries. But the boolean monadic restriction prevents us from computing an index from another index as a subroutine! Even worse, it seems impossible to count the number of blank indices before a certain point.

What we have to do instead is modify $\pi$ so that it ``ignores'' $\square$ characters. Whenever $\pi$ tests whether an index is the min or the max, $\pi_b$ tests whether the same index is the min or max \emph{non-blank} index. Whenever $\pi$ takes the successor or predecessor of an index, $\pi_b$ goes forward or backward to the next non-blank index. Here, the limited recursive capabilities of a BMRS are enough: all we have to remember are what state we're in and what direction we're going, and use tail recursion to skip over any blank indices we see. In so doing, we avoid having to count anything.

So much for the overview, let us see how it works more formally. For the remainder of this section, let $\pi : \Sigma_\square \times m \to \Gamma$ be a fixed strict interpretation.

\subsection{A normal form}
We can ensure that all calls to $\max$ and $\min$ in $\pi$ have the form $\max(\x)$ and $\min(\x)$ respectively. This construction is rather artificial: simply make two additional recursive functions $\f_{\max}(\x) \equiv \max(\x)$ and $\f_{\min}(\x) \equiv \min(\x)$, and replace all other occurrences of $\max(T)$ or $\min(T)$ by $\f_{\max}(T)$ and $\f_{\min}(T)$ respectively. In precisely the same way, we can ensure that all calls to any character $\gamma \in \Gamma$ have the form $\gamma(\x)$.

Next, we can ensure that every recursive call has one of three forms: $\f(\x)$, $\f(\suc \x)$, and $\f(\pred \x)$, for some recursive function name $\f$. Due to the restrictions of BMRS syntax, it's already the case that each recursive call is of the form $\f(T)$, for some index-valued term $T$, and each index-valued term is a string of $\suc$'s and $\pred$'s applied to $\x$. We can reduce this to the above three forms by adding more recursive functions. For example, $\f(\suc \suc \x)$ can be replaced by $\mathtt{g}(\suc \x)$, where $\mathtt{g}(\x) \equiv \f(\suc \x)$.

So what we have done is shown that we can assume every call to a boolean-valued primitive ($\max$, $\min$, or some $\gamma$) in $\pi$ is of the form $\max(\x)$, $\min(\x)$, and $\gamma(\x)$, and every call to a recursive function $\f$ in $\pi$ is of the form $\f(\x)$, $\f(\suc \x)$, or $\f(\pred \x)$.

\subsection{A program transformation}
We now define a program transformation $p \mapsto p^\star$ from headless $\Gamma$-BMRS's to headless $\Gamma_\square$-BMRS's; the rough idea being that for any string $t \in \Gamma^\star$, $p^\star$ will do on any padding of $t$ what $p$ does on $t$.

In the definitions below, by a ``blank'' or ``non-blank'' index of a given string, we simply mean an index which does or does not carry the character $\square$ respectively.

\begin{definition}
For any string $s \in \Gamma_\square^\star$, let $x \mapsto x^\star$ be the map from non-blank indices of $s$ to indices of $d(s)$ defined by $x^\star = x - \delta$, where $\delta$ is the number of blank predecessors of $x$ in $s$.

Extend $x \mapsto x^\star$ by two partial maps $x \mapsto x^S$ and $x \mapsto x^P$ from \emph{all} indices of $s$ to indices of $d(s)$. Namely, $x^S = y^\star$, where $y$ is the least non-blank index $\ge x$, and $x^P = z^\star$, where $z$ is the greatest non-blank index $\le x$.\footnote{These maps may be partial as $x^S$ and $x^P$ will not be well-defined if there is no succeeding or preceding non-blank index respectively.}
\end{definition}

For example, if $\Gamma = \{a,b\}$ and $s = a \square \square b \square a$, then $d(s) = aba$ and $0^\star = 0$, $3^\star = 1$, and $5^\star = 2$. Moreover $0^P = 1^P = 2^P = 0$, $3^P = 1$, and $4^P = 5^P = 2$. Finally $0^S = 0$, $1^S = 2^S = 3^S = 1$, and $4^S = 5^S = 2$. Note that all of these maps are non-decreasing and that $x \mapsto x^\star$ is a bijection between non-blank indices of $s$ and indices of $d(s)$. 

Now consider the function $\maxc$ defined by:
\begin{align*}
    \maxc(\x) &\equiv \ite{\square(\x)}{\f(\x)}{\false} \\
    \f(\x) &\equiv \ite{\max(\x)}{\true}{\ite{\square(\suc \x)}{\false}{\f(\suc \x)}}.
\end{align*}
Then $\maxc$ detects the maximum non-blank index. We can similarly define $\minc$.

The following remark encapsulates several basic properties of the maps defined so far:

\begin{rem}\label{rem:properties of blanked out string}
For any string $s \in \Gamma_\square^\star$, any non-blank index $x$ of $s$ and any character $\gamma \in \Gamma$, we have:
\begin{enumerate}
    \item  $s \models \gamma(x) \iff d(s) \models \gamma(x^\star),$
    \item  $s \vdash \maxc(x) \to \true \iff d(s) \models \max(x^\star),$
    \item $s \vdash \minc(x) \to \true \iff d(s) \models \min(x^\star),$
    \item $ s \vdash \maxc(x) \to \false \iff d(s) \models \neg \max(x^\star),$
    \item $ s \vdash \minc(x) \to \false \iff d(s) \models \neg \min(x^\star),$
    \item $x^S = x^P = x^\star$, and
    \item $\pred x^\star = (\pred x)^P$, and $\suc x^\star = (\suc x)^S $.
\end{enumerate}
Finally, for any \emph{blank} index $x$ of $s$, $x^S = (\suc x)^S$ and $x^P = (\pred x)^P$.
\end{rem}

\begin{definition}\label{def:T star}
For each recursive function name $\f$, let $\f^\star$, $\f^S$, and $\f^P$ be three distinct recursive function symbols.
We define a transformation $T \mapsto T^\star$ from boolean-valued $\Gamma$-terms (in the above normal form) to boolean-valued $\Gamma_\square$-terms as follows.
\begin{itemize}
    \item If $T \equiv \true$ or $T \equiv \false$, then $T^\star \equiv T$.

    \item If $T \equiv \gamma(\x)$ for any character $\gamma \in \Gamma$, then $T^\star \equiv \gamma(\x)$.
    
    \item If $T \equiv \max(\x)$ or $T \equiv \min(\x)$, then $T^\star \equiv \maxc(\x)$ or $T \equiv \minc(\x)$ respectively.
    
    \item If $T \equiv \f(\x)$, then $T^\star \equiv \f^\star(\x)$, for any recursive function name $\f$. 
    
    \item If $T \equiv \f(\suc \x)$, then $T^\star \equiv \f^S(\suc \x)$, for any recursive function name $\f$.
    
    \item If $T \equiv \f(\pred \x)$, then $T^\star \equiv \f^P(\pred \x)$, for any recursive function name $\f$.
    
    \item If $T \equiv \ite{T_0}{T_1}{T_2}$ then $T^\star \equiv \ite{T_0^\star}{T_1^\star}{T_2^\star}$.
\end{itemize}
\end{definition}

\begin{definition}
For any headless $\Gamma$-BMRS $\p = (\f_i(\x) = T_i)_{0 \le i \le k}$, define the $\Gamma_\square$-BMRS $\p^\star$ by $(\f_i^\star(\x) = T_i^\star)_{0 \le i \le k} $ plus, for each $0 \le i \le k$, 
$$ \f_i^S(\x) = \ite{\square(\x)}{\f_i^S(\suc \x)}{T_i^\star},$$
$$ \f_i^P(\x) = \ite{\square(\x)}{\f_i^P(\pred \x)}{T_i^\star}.$$
\end{definition}

\begin{theorem}\label{thm:correctness of star transf}
For any string $s \in \Gamma_\square^\star$, headless $\Gamma$-BMRS $p$, non-blank index $x$ of $s$, boolean-valued $p$-term $T$, and boolean $b$,
$$ d(s), x^\star \vdash_\p T \to b \implies s, x \vdash_{\p^\star} T^\star \to b.$$
Furthermore, for any index $x$ of $s$, any recursive function symbol $\f$, and any boolean $b$, if $x^S$ exists, then 
$$d(s), x^S \vdash_\p \f(\x) \to b \implies s,x \vdash_{\p^\star} \f^S(\x) \to b,$$ and if $x^P$ exists, then
$$ d(s),x^P \vdash_\p \f(\x) \to b \implies s,x \vdash_{\p^\star} \f^P(\x) \to b.$$
\end{theorem}

\begin{proof}
Fix $s$, $p$, $x$, $T$ and $b$. Let $t = d(s)$ be the underlying $\Gamma^\star$-string of $s$.
For legibility we omit the subscripts under the turnstiles. We understand $t,\dots \vdash \dots$ to mean $t, \dots \vdash_p \dots$ and $s , \dots \vdash \dots$ to mean $s, \dots \vdash_{p^\star} \dots$. We prove all three statements by simultaneous induction on the length of the computation in $p$.

First, let us prove $t,x^\star \vdash T \to b \implies s,x \vdash T^\star \to b.$ Assume $t,x^\star \vdash T \to b$.
If $T \equiv \gamma(\x)$ for some character $\gamma \in \Gamma$, then $T^\star \equiv \gamma(\x)$, and the conclusion follows from Remark \ref{rem:properties of blanked out string}. If $T \equiv \max(\x)$ or $\min(\x)$ then $T^\star \equiv \maxc(\x)$ or $\minc(\x)$ and the conclusion again follows from Remark \ref{rem:properties of blanked out string}.

Suppose that $T \equiv \f_i(\x)$, and let $\f_i(\x) = T_i$ be the recursive definition of $\f_i$ in $\p$. Then $t,x^\star \vdash T_i \to b$, and by induction, $s,x \vdash T_i^\star \to b$. But $\f_i^\star(\x) = T_i^\star$ is the recursive definition of $\f_i^\star$ in $\p^\star$, so $s,x \vdash \f_i^\star(\x) \to b$, which is what we want to prove, because $T^\star \equiv \f_i^\star(\x)$.

Suppose that $T \equiv \f(\pred \x)$. (The case that $T \equiv \f(\suc \x)$ is similar.) Since $t,x^\star \vdash \f(\pred \x) \to b$, $t, \pred x^\star \vdash \f(\x) \to b$. Furthermore suppose that $\pred x$ is a non-blank index of $s$. By Remark \ref{rem:properties of blanked out string} $t \models \pred x^\star = (\pred x)^P$.
Therefore, $t,(\pred x)^P \vdash \f(\x) \to b$. By induction, $s, \pred x \vdash \f^P(\x) \to b$, thus $s, x \vdash \f^P(\pred \x) \to b$, which is what we wanted to show, since $T^\star \equiv \f^P(\pred \x)$.

Now suppose that $y = \pred x$ is a blank index of $s$. Then $t \models \pred x^\star = y^P$ (Remark \ref{rem:properties of blanked out string}). Then $t, y^P \vdash \f(\x) \to b$; by induction, $s,y \vdash \f^P(\x) \to b$, so $s, x \vdash \f^P(\pred \x) \to b $, which is what we wanted to show.

Finally, suppose that $T \equiv \ite{T_0}{T_1}{T_2}$ and $t,x^\star \vdash T_0 \to \top$. (The case $t,x^\star \vdash T_0 \to \bot$ is similar, replacing $T_1$ by $T_2$.) Then $t,x^\star \vdash T_1 \to b$, so $s,x \vdash T_0^\star \to \top$, and $s,x \vdash T_1^\star \to b$. Hence, $s,x \vdash T^\star \to b$. This concludes the proof of the first statement.

Next, let us prove that $t,x^S \vdash \f(\x) \to b \implies s,x \vdash \f^S(\x) \to b$.
(The proof of $t,x^P \vdash \f(\x) \to b \implies s,x \vdash \f^P(\x) \to b$ is similar.)
Assume $t,x^S \vdash \f(\x) \to b$.

Suppose first that $x$ is a non-blank index, so $x^S = x^\star$. If $\f(\x) \equiv \f_i(\x)$, then $t,x^\star \vdash T_i \to b$. By induction, $s,x \vdash T_i^\star \to b$. Therefore
$$ s,x \vdash \ite{\square(\x)}{\f_i^S(\suc \x)}{T_i^\star} \to b,$$
which means that $s,x \vdash \f_i^S(\x) \to b$.

Finally, suppose that $x$ is the index of a blank character. Then $t \vdash (\suc x)^S = x^S$ by Remark \ref{rem:properties of blanked out string} and hence $t,(\suc x)^S \vdash \f(\x) \to b$. By induction, $s,\suc x \vdash \f^S(\x) \to b$. Therefore,
$$ s,x \vdash \ite{\square(\x)}{\f^S(\suc \x)}{\f(\x)} \to b, $$
which means $s,x \vdash \f^S(\x) \to b$. This concludes the proof.
\end{proof}

\subsection{The definition of $\pi_b$}
Finally, we are in a position to define $\pi_b$ from $\pi$. Recall that $\pi$ is a strict interpretation of type $\Sigma_\square \times m \to \Gamma$.

\begin{definition}
Define the interpretation $\pi_b : \Sigma_\square \times m \to \Gamma_\square$ as follows:
\begin{itemize}
    \item The body of $\pi_b$ is $p^\star$, where $p$ is the body of $\pi$.
    
    \item For each character $\sigma \in \Sigma$ and $i < m$, the head $\pi_b(\sigma,i)$ is
    $$ \ite{\square(\x)}{\false}{\pi(\sigma,i)^\star} .$$
    
    \item For each $i < m$, the head $\pi_b(\square,i)$ is
    $$ \ite{\square(\x)}{\true}{\pi(\square,i)^\star}.$$
\end{itemize}
\end{definition}

\begin{lemma}
$\pi_b$ is a strict interpretation
\end{lemma}
\begin{proof}
Fix a string $s \in \Gamma_\square^\star$, $x < |s|$, and $i < m$. If $s_x = \square$ then $s,x \vdash \pi_b(\square,i) \to \true$ and $s,x \vdash \pi_b(\sigma,i) \to \false$ for every $\sigma \in \Sigma$.

Otherwise $x$ is a non-blank index of $s$. By strictness of $\pi$, there is a unique $\sigma \in \Sigma_\square$ such that $d(s) , x^\star \vdash \pi(\sigma,i) \to \true$, and for every other $\tau \in \Sigma_\square$, $d(s), x^\star \vdash \pi(\tau,i) \to \false$. By Theorem \ref{thm:correctness of star transf}, $s,x \vdash \pi_b(\sigma,i) \to \true$, and $s,x \vdash \pi_b(\tau,i) \to \false$ for every other $\tau$. This proves strictness of $\pi_b$.
\end{proof}

Finally, we prove correctness.

\begin{theorem}\label{thm:correctnes of pi_b}
For any string $s \in \Gamma_\square^\star$, $d(\sempar{\pi_b}(s)) = d(\sempar{\pi}(d(s)))$.
\end{theorem}

\begin{proof}
Fix a string $s \in \Gamma^\star_\square$ and let $t = d(s)$ be the underlying string in $\Gamma^\star$. Let $u = \sempar{\pi_b}(s)$ and $v = \sempar{\pi}(t)$, so that $|u| = m|s|$ and $|v| = m|t|$; we want to show that $d(u) = d(v)$. Instead of showing this directly, we show that $u$ can be obtained from $v$ by padding it with more $\square$'s, which comes to the same thing.

Let $X \subseteq |s|$ be the set of non-blank indices of $s$, so that $|X| = |t|$ and $x \mapsto x^\star$ is a bijection $X \to |t|$. Let $Y = \{y < m|s| : y \div m \in X\}$; here $\div$ refers to integer division. Then $Y$ can be identified as a set of indices of $u$; moreover $|Y| = |v|$. Define the map $f : Y \to |v|$ by $f(mq + r) = mq^\star + r$ for every $q \in X$ and $r < m$. Then $f$ is the unique monotone bijection between $Y$ and indices of $v$.

Let $y$ be an index of $u$ not in $Y$. Let $q$ and $r$ be the quotient and remainder of $y \div m$ respectively. Then $q \notin X$, so $s \models \square(q)$. Therefore $s,q \vdash \pi_b(\square,r) \to \true$, so $\sempar{\pi_b}(s) \vdash \square(y)$. In other words, $\sempar{\pi_b}(s)$ carries $\square$ on any index outside $Y$.

Now let $y$ be an index of $u$ in $Y$, and again, let $q$ and $r$ be the quotient and remainder upon division by $m$. Let $z = f(y) = mq^\star + r$ and let $\sigma$ be the character of $v$ at index $z$, so that $\sempar{\pi}(t) \models \sigma(z)$. Then $t,q^\star \vdash \pi(\sigma,r) \to \true$, and by Theorem \ref{thm:correctnes of pi_b}, $s,q \vdash \pi(\sigma,r)^\star \to \true$. By definition of $\pi_b$, $s,q \vdash \pi_b(\sigma,r) \to \true$; hence, $\sempar{\pi_b}(s) \models \sigma(y)$. In other words, for every $y \in Y$, the character of $\sempar{\pi_b}(s)$ at $y$ is the same as the character of $\sempar{\pi}(t)$ at $f(y)$.

What we have done is partitioned the indices of $u$ into $Y$ and $|u| \setminus Y$. In the former part, $v$ appears as a substring; in the latter part, we have only blank characters. Therefore, $u$ is obtained from $v$ by padding it with blanks, and hence $d(u) = d(v)$.
\end{proof}

\section{BMRS and the rational functions}\label{s:rational}

We can now turn to the main results of the paper. 
First, we define a general composition operation for BMRS interpretations that are not necessarily strict.
Note  the overloading of $\otimes$ for both strict- and non-strict interpretations. 

\begin{definition}
  \label{def:comp}
  For two order-preserving BMRS interpretations $\rho : \Delta \times n \to \Sigma$ and $\pi : \Sigma \times m \to \Gamma$, let 
  \[\rho\otimes\pi=(\rho_b'\otimes\pi')^\dagger,\]
  where $\rho'_b$ is taken as $(\rho')_b$.
\end{definition}

Then we have

\begin{theorem}
  \label{thm:comp}
  For any $\rho$ and $\pi$ as above, $\sempar{\rho\otimes\pi}=\sempar{\rho}\circ\sempar{\pi}$. 
\end{theorem}

\begin{proof}
By Lemma \ref{lem:destrictification correctness}, $\sempar{\rho \otimes \sigma} = d \circ \sempar{\rho'_b \otimes \pi'}$. By Lemma \ref{lem:strict comp correctness}, $\sempar{\rho'_b \otimes \pi'} = \sempar{\rho'_b} \circ \sempar{\pi'}$. By Theorem \ref{thm:correctnes of pi_b}, $d \circ \sempar{\rho'_b} = d \circ \sempar{\rho'} \circ d$. Hence
$$ \sempar{\rho \otimes \sigma} = d \circ \sempar{\rho'} \circ d \circ \sempar{\pi'},$$ but by Lemma \ref{lem:strictification correctness}, $d \circ \sempar{\rho'} = \sempar{\rho}$ and $d \circ \sempar{\pi'} = \sempar{\pi}$. Hence $\sempar{\rho \otimes \sigma} = \sempar{\rho} \circ \sempar{\sigma}$.
\end{proof}

Finally we obtain our main characterization of rational functions by order-preserving BMRS interpretations. We crucially use a theorem of Elgot and Mezei that every rational function can be decomposed into a composition of a left- with a right-subsequential function.

\begin{theorem}[\cite{ElgotMezei}]
  \label{thm:em}
  For every rational function $f$, $f=g\circ h$ for some left-subsequential function $g$ and some right-subsequential function $h$. 
\end{theorem}

\begin{theorem}
  For any well-defined order-preserving $\mathrm{BMRS}$ interpretation $\pi$, $\sempar{\pi}$ is a rational function. 
  Likewise, given a rational function $f$, $f=\sempar{\pi}$ for some order-preserving $\mathrm{BMRS}$ interpretation $\pi$. 
\end{theorem}

\begin{proof}
  The forward direction is immediate from the fact that BMRS are a fragment of MSO, plus the fact that order-preserving MSO interpretations define rational functions \cite{bojanczyk14,filiot15}. In the backwards direction, given any rational function $f$, consider its decomposition as $g \circ h$ guaranteed by Theorem \ref{thm:em}. By Theorem \ref{thm:subseq} there are order-preserving BMRS interpretations $\rho$ and $\pi$ such that $g = \sempar{\rho}$ and $h = \sempar{\pi}$. Finally, $f = \sempar{\rho \otimes \pi}$ by Theorem \ref{thm:comp}.
\end{proof}

\section{Discussion and open questions}
\label{s:conclusion}

We factored composition of order-preserving interpretations into four sub-problems: strictification, de-strictification, strict composition, and blank enrichment. The solutions to the first three are very general and can be replicated in almost any logic. Strictification relies only on closure under boolean operations, strict composition relies on compositionality (syntactically, the ability to substitute terms for variables of like type), and de-strictification relies on nothing at all.

Syntactic blank enrichment, on the other hand, seems to call upon all of the limited recursive power afforded by BMRS. This suggests to us that it is an important problem to examine in the context of order-preserving interpretations over different logics. Moreover, our solution seems to adapt to the following slightly more general problem:

\smallskip
\begin{minipage}{0.55\textwidth}
Let's say we are given a relation $L$ which selects a subset of indices of any given string. Now we have two (strict) functions $f$ and $g$, and we want to ``hybridize" them with respect to $L$ so that, on a given string $s$, we apply $f$ ``on'' $L$ and $g$ ``outside" $L$. For example if $s = ababb$ and $L$ selects indices 0, 1, and 4 of $s$, then $s$ restricted to $L$ is $abb$ and $s$ restricted to its complement is $ab$. If $f(abb) = 101$ and $g(ab) = 22$, then the $L$-hybrid of $f$ and $g$ applied to $s$ is $10221$. (See figure.)
\end{minipage}
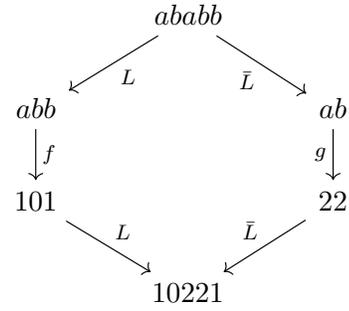
\begin{minipage}{0.4\textwidth}
    \[\begin{tikzcd}
	& ababb \\
	abb & {} & ab \\
	101 && 22 \\
	& 10221
	\arrow["L", from=1-2, to=2-1]
	\arrow["{\bar{L}}"', from=1-2, to=2-3]
	\arrow["f", from=2-1, to=3-1]
	\arrow["g"', from=2-3, to=3-3]
	\arrow["L", from=3-1, to=4-2]
	\arrow["{\bar{L}}"', from=3-3, to=4-2]
\end{tikzcd}\]
\end{minipage}
\smallskip

Blank enrichment is simply the special case of hybridization when $L$ is $\square(\x)$, $f$ is the given function, and $g$ is the constant-$\square$ function. It seems like the technique for blank enrichment will readily generalize to the more general problem. Many other natural functionals can be expressed in terms of hybridization, for example \emph{concatenation:} given two functions $f,g:\Sigma^\star \to \Gamma^\star$, define the function $\mathit{Concat}(f,g) : \Sigma^\star_\square \to \Gamma^\star_\square$ by $ u\square v \mapsto f(u) \square g(v)$, at least on strings with exactly one $\square$.

Concatenation-like functionals are an important example of functionals (or \emph{combinators}) used in algebraic characterizations of function classes such as the regular functions \cite{AFR14}. Realizing combinators syntactically by program transformations (as we have done here for rational functions) gives us a path towards capturing results different from---and arguably cleaner than---compiling logical interpretations into transducers and vice versa.

A clear next step is to see whether there is a BMRS characterization of the class of regular functions. Since these are captured by MSO interpretations (not necessarily order-preserving), one might guess that the same holds for BMRS. However it is not exactly clear what a non-order preserving BMRS interpretation might be: we would need to define the successor function on the output string, which means accommodating non-boolean valued recursive functions in our programs.

Another open direction involves adapting BMRS to other data types, like trees or graphs, which admit some notion of regularity. Here again we encounter the foundational problem of how to extend the boolean monadic paradigm to data types which do not fit neatly into it. The advantage of order-preserving BMRS interpretations over strings lies in their simplicity combined with their intensional expressiveness. We would regard these qualities as an acid test of any proposed extension of BMRS to other domains.

\bibliographystyle{alphaurl}
\bibliography{references}

\end{document}